\documentclass[11pt,DIV10,final,a4paper,USenglish]{scrartcl}
\usepackage{etex}
\usepackage[utf8]{inputenc}%(%\usepackage{ngerman}
\usepackage{amssymb,amsmath,amsthm}
\usepackage{fancyhdr}
\usepackage{graphicx}
\usepackage{listings}
\usepackage[all]{xy}
\usepackage{stmaryrd}
\usepackage{pgfpages}
\usepackage[final]{hyperref}
\usepackage{xspace}
\usepackage{tikz}
\usetikzlibrary{trees}
\usepackage{multicol}

\usepackage{prettyref}
\newcommand{\prref}[1]{\prettyref{#1}}
\newrefformat{thm}{Theorem~\ref{#1}}
\newrefformat{lem}{Lemma~\ref{#1}}
\newrefformat{def}{Definition~\ref{#1}}
\newrefformat{cor}{Corollary~\ref{#1}}
\newrefformat{prop}{Proposition~\ref{#1}}
\newrefformat{sec}{Section~\ref{#1}}
\newrefformat{kap}{Chapter~\ref{#1}}
\newrefformat{ex}{Example~\ref{#1}}
\newrefformat{eq}{Equation~(\ref{#1})}
\newrefformat{rem}{Remark~\ref{#1}}
\newrefformat{fig}{Figure~\ref{#1}}
\newrefformat{par}{Paragraph~\ref{#1}}

\colorlet{DMnormalbackcolor}{gray!25}
\colorlet{DMlightbackcolor}{gray!10}
\colorlet{DMmediumbackcolor}{gray!20}
\colorlet{DMdarkbackcolor}{gray!30}

\definecolor{blue}{rgb}{0.211,0.211,0.656}

\definecolor{dgreen}{rgb}{0.0,0.4,0.0}

\newtheorem{theorem}{{\bf Theorem}}[]
\newtheorem{corollary}[theorem]{{\bf Corollary}}

\newtheorem{example}[theorem]{{\bf Example}}
\newtheorem{lemma}[theorem]{{\bf Lemma}}
\theoremstyle{plain}
\newtheorem{proposition}[theorem]{{\bf Proposition}}

\newcommand{\refthm}[1]{Theorem~\ref{#1}}

\newcommand{\reflem}[1]{Lemma~\ref{#1}}
\newcommand{\refprop}[1]{Proposition~\ref{#1}}

\newcommand{\reffig}[1]{Figure~\ref{#1}}

\newcommand{\refex}[1]{Example~\ref{#1}}

\newcommand{\set}[2]{\left\{#1\mathrel{\left|\vphantom{#1}\vphantom{#2}\right.}#2\right\}}
\newcommand{\oneset}[1]{\left\{\mathinner{#1}\right\}}
\newcommand{\smallset}[1]{\left\{\mathinner{#1}\right\}}
\newcommand{\os}[1]{\oneset{#1}}

\newcommand{\arc}[1]{\overset{#1}\ra}
\newcommand{\abs}[1]{\left|\mathinner{#1}\right|}

\newcommand{\dbracket}[1]{\left\llbracket \mathinner{#1} \right\rrbracket}

\newcommand{\N}{\mathbb{N}}

\renewcommand{\phi}{\varphi}

\newcommand{\oo}{\omega}
\newcommand{\alp}{\alpha}
\newcommand{\bet}{\beta}
\newcommand{\gam}{\gamma}

\newcommand{\lam}{\lambda}
\newcommand{\sig}{\sigma}

\newcommand{\cH}{\mathcal{H}}
\newcommand{\cC}{\mathcal{C}}

\newcommand{\bfAb}{\mathbf{Ab}}
\newcommand{\bfSol}{\mathbf{Sol}}
\newcommand{\bfI}{\mathbf{1}}

\newcommand{\bfH}{\mathbf{H}}

\newcommand{\bfU}{\mathbf{U}}
\newcommand{\bfG}{\mathbf{G}}

\newcommand{\bfV}{\mathbf{V}}

\newcommand{\ov}[1]{\overline{#1}}
\newcommand{\oi}[1]{{#1}^{-1}}

\newcommand\ra{\longrightarrow}

\newcommand{\wt}[1]{\widetilde{#1}}

\newcommand\subsem{subsemigroup\xspace}

\newcommand{\eg}{eg.\xspace}
\newcommand{\IFF}{if and only if\xspace}
\newcommand{\homo}{homomorphism\xspace}

\newcommand{\vars}{varieties\xspace}

\newcommand{\FO}{\mathrm{FO}}
\newcommand{\LTL}{\mathrm{LTL}}
\newcommand{\SF}{\mathrm{SF}}
\newcommand{\AP}{\mathbf{Ap}}

\newcommand{\schuetz}{Schützenberger\xspace}

\newcommand{\sse}{\subseteq}
\newcommand{\es}{\emptyset}
\newcommand{\sm}{\setminus}

\newcommand{\svarietyfont}[1]{\ensuremath{\mathcal{#1}}\xspace}
\newcommand{\SD}{\ensuremath{\mathrm{SD}}}
\newcommand{\SDG}[1]{\ensuremath{\mathrm{SD}_\svarietyfont{#1}}}
\newcommand{\SDH}[1]{\ensuremath{\mathrm{SD}_{#1}}}
\newcommand{\varietyfont}[1]{\ensuremath{\mathbf{#1}}\xspace}
\newcommand{\varietyG}{\varietyfont{G}}
\newcommand{\varietyH}{\varietyfont{H}}
\newcommand{\varietyHline}{\overline{\varietyfont{H}}}

\newcommand{\Synt}{\ensuremath{\mathrm{Synt}}}

\newcommand{\BRExp}[1]{\ensuremath{\mathrm{Exp(#1)}}}
\newcommand{\LocRees}[1]{\ensuremath{\mathrm{LocRees(#1)}}}
\newcommand{\Rees}[1]{\ensuremath{\mathrm{Rees(#1)}}}

\title{Characterizing classes of regular languages using prefix codes of bounded synchronization delay}

\author{Volker Diekert and Tobias Walter\footnote{Supported by the German Research Foundation (DFG) under grant DI 435/6-1.}}
\publishers{University of Stuttgart, FMI\\
	Universit\"atsstra{\ss}e 38, 70569 Stuttgart, Germany\\
	\texttt{\{diekert,walter\}@fmi.uni-stuttgart.de}}
\date{}
\begin{document}
\maketitle
\begin{abstract}
In this paper we continue a classical work of Sch\"utzenberger on codes with bounded synchronization delay. He was interested to characterize those regular languages where the groups in the syntactic monoid belong to a variety $\bfH$. He allowed operations on the language side which are union, intersection, concatenation and modified Kleene-star involving 
a mapping of a prefix code of bounded synchronization delay to a group $G\in \bfH$, but no complementation. 
In our notation this leads to the language classes $\SDH{G}(A^\infty)$ and $\SDG{\bfH}(A^\infty)$. Our main result  shows that $\SDG{\bfH}(A^\infty)$ always corresponds to the languages having syntactic monoids where all subgroups are in $\bfH$. 
Sch\"utzenberger showed this 
for a variety $\bfH$ if $\bfH$ contains Abelian groups, only. 
Our method shows the general result for all $\bfH$ directly on finite and infinite words. Furthermore, we introduce the notion of \emph{local Rees products} which refers to a simple type of classical Rees extensions. We give a decomposition of a monoid in terms of its groups and local Rees products. This gives a somewhat similar, but simpler decomposition than in  Rhodes' synthesis theorem. Moreover, we need a singly exponential number of operations, only. Finally, our decomposition yields an answer to a question in a recent paper of Almeida and Kl{\'i}ma about varieties that are closed under Rees products.
\end{abstract}
%Zitate: Done
\section{Introduction}\label{sec:intro}
A fundamental result of Sch\"utzenberger characterizes the class of star-free languages $\SF$ as exactly those languages which are group-free, that is, aperiodic \cite{sch65sf}. One usually abbreviates this result by $\SF = \AP$. Sch\"utzenberger also found another, but less prominent characterization of $\SF$: the star-free languages are exactly the class of languages which can be defined inductively by finite languages and closure under union, concatenation, and the Kleene-star restricted to prefix codes of bounded synchronization delay \cite{Schutzenberger1975d}. 
This result is abbreviated by $\AP = \SD$. It is actually stronger than the famous $\SF = \AP$ because $\SD \sse \SF \sse \AP$ is easy, so  
$\SF = \AP$ follows directly from $\AP \sse \SD$.
The result $\AP = \SD$ has been extended to infinite words first in \cite{DiekertKufleitner14tocs}. The extension to infinite words became possible thanks to a ``local divisor approach'', which also is a main tool in this paper. 

Sch\"utzenberger did not stop by showing $\AP = \SD$. In retrospective he started a program: in \cite{Schutzenberger1974b} he was able to prove an analogue of $\AP = \SD$ for languages where syntactic monoids have Abelian subgroups, only. In our notation $\AP = \SD$ means
$\ov \bfI(A^\infty)= \SDH{\bfI}(A^\infty)$; and the main result in \cite{Schutzenberger1974b} is ``essentially'' equivalent to  $\ov \bfAb(A^*) = \SDH{\bfAb}(A^*)$. (We write ``essentially'' because using the structure theory of Abelian groups, a sharper version than $\ov \bfAb(A^*) = \SDH{\bfAb}(A^*)$ is possible.) 
The proofs \cite{Schutzenberger1974b} use deep results in semigroup theory; and no such result beyond Abelian groups was known so far.
Our result generalizes $\ov \bfAb(A^\infty)= \SDH{\bfAb}(A^\infty)$ to every variety $\bfH$ of finite groups: we show $\ov \bfH(A^\infty)= \SDH{\bfH}(A^\infty)$. We were able to prove it with much less technical machinery compared to \cite{Schutzenberger1974b}. For example, no knowledge in Krohn-Rhodes theory is required.  

Actually, our result is a generalization of $\ov \bfAb(A^*) = \SDH{\bfAb}(A^*)$ \cite{Schutzenberger1974b} 
and also of $\AP(A^\infty) = \SDH{}(A^\infty)$ \cite{DiekertKufleitner14tocs}. More precisely, we give a characterization of languages which are recognized by monoids where all subgroups belong to ${\bfH}$. The characterization uses an inductive scheme
starting with all finite subsets of finite words, allows concatentation, union, no(!) complementation, but a restricted use of a generalized Kleene-star (and $\oo$-power in the case of infinite words). Let us explain the 
\emph{generalized Kleene-star} in our context.
Instead of putting the star above a single language, consider first a disjoint union $K=\bigcup\set{K_g}{g\in G}$ where $G$ is a finite group and each $K_g$ is regular in $A^*$. The ``generalized star'' associates with such a disjoint union the following language: 
 $$\set{u_{g_1}\cdots u_{g_k}\in K^*}{u_{g_i}\in K_{g_i}\wedge g_1 \cdots {g_k}= 1 \in G}.$$ 
Clearly, we obtain a regular language, but without any restriction, allowing such a ``general star'' yields all regular languages,  even in the case of the trivial group. So, the construction is of no interest without a simultaneous restriction. The restriction considered in 
\cite{Schutzenberger1974b} yields an inductive scheme to define a class $\cC$. The restriction says 
that such a generalized Kleene-star is allowed only over a disjoint union  $K=\bigcup\set{K_g}{g\in G}$ where each $K_g$ already belongs to $\cC$ and where $K$ is, in addition, a prefix code of bounded synchronization delay. The initials in ``synchronization delay'' led to the notation $\SD$; and an indexed version $\SDH{G}$ (resp.~$\SDH{\bfH}$) 
refers to ``synchronization delay over $G$'' (resp.~over a finite group in ${\bfH}$). Since we also deal with infinite words we apply the same restriction 
to $\oo$-powers. 

Our results give also a new characterization for various other classes. 
For example, by a result of Straubing, Th{\'e}rien and Thomas \cite{stt95IC}, the class of languages, % $\ov \bfSol$, which is the class that their 
having syntactic monoids where all subgroups are solvable, coincides with $(\FO+\mathrm{MOD})[<]$. Here, $(\FO+\mathrm{MOD})[<]$ means the class of languages defined by the logic $(\FO+\mathrm{MOD})[<]$. Thus, we are able to give a new language characterization: 
$(\FO+\mathrm{MOD})[<](A^\infty)= \SDH{\bfSol}(A^\infty).$

Moreover, as a sort of byproduct of $\ov \bfH = \SDH{\bfH}$, we obtain a simple and purely algebraic characterization of the monoids in $\overline{\bfH}$. Every monoid in $\overline{\bfH}$ can be decomposed in at most exponentially many iterated Rees products of groups in $\bfH$. The iteration uses only a very restricted version of Rees extensions: \emph{local Rees products}. This means
we obtain every finite monoid which is not a group as a divisor of a Rees extension between two proper divisors of $M$, one of them a proper submonoid, the other one a ``local divisor''.  

Our decomposition result is similar to the synthesis theory of Rhodes and Allen \cite{RhodesA73}. Moreover, our technique gives a singly exponential bound on the number of operations whereas no such bound was known by \cite{RhodesA73}.
Finally, using this decomposition, we answer a recent question of Almeida and Kl{\'i}ma \cite{AlmeidaK16} concerning varieties which are closed under Rees products. 

\section{Preliminaries}\label{sec:prelim}
Throughout, $A$ denotes a finite alphabet and $A^*$ is the free monoid over $A$. It consists of all finite words. The empty word is denoted by $1$ as the neutral elements in other monoids or groups. %(unless we use an additive notation for the multiplication). 
The set of non-empty finite words is $A^+$; it is the free semigroup over $A$. By $A^\omega$ we denote the set of all infinite words with letters in $A$. 
For a set $K\subseteq A^*$, we let $K^\omega = \set{u_1u_2\cdots }{u_i \in K \text{ non-empty}, i\in \mathbb N} \subseteq A^\omega$. In particular, $K^\omega = (K\setminus \os{1})^\omega$.
Since our results concern finite and infinite words, it is convenient to treat finite and infinite words simultaneously. We define $A^\infty = A^* \cup A^\omega$ to be the set of finite or infinite words. Accordingly, a  \emph{language} $L$ is a subset of $A^\infty$. We say 
that $L$ is \emph{regular}, if  first, $L\cap A^*$ is regular and second, 
$L\cap A^\oo$  is $\oo$-regular in the standard meaning of formal language theory. % \cite{?}.
In order to study regular languages algebraically, one considers finite monoids. 
A \emph{divisor} of a monoid $M$ is a monoid $N$ which is a homomorphic image of a \subsem of $M$. In this case we write $N \preceq M$.
A \subsem $S$ of $M$ is in our setting a divisor \IFF $S$ is a monoid (but not necessarily a submonoid of $M$). 
%This does not mean that $S$ is a submonoid of $M$ as ....
A \emph{variety} of finite monoids -- hence, in Birkhoff's setting:  a \emph{pseudovariety} -- is a class of finite monoids $\bfV$ which is closed under finite direct products and under division:
\begin{itemize}
\item 
If $I$ is a finite index set and $M_i \in \bfV$ for each $i\in I$, then $\prod_{i\in I} M_i \in \bfV$. In particular, the trivial group $\os 1$ belongs to $\bfV$. 
\item 
If $M \in \bfV$ and $N \preceq M$, then $N \in \bfV$.
\end{itemize}

Classical formal language theory states ``regular'' is the same as ``recognizable''. This means: $L\sse A^*$  is regular \IFF its syntactic 
monoid is finite; $L\sse A^\oo$ is regular \IFF  its syntactic monoid (in the sense of Arnold) is finite and, in addition, $L$ is saturated by the syntactic congruence, see \eg \cite{pp04,tho90handbook}. Here we use a notion of recognizability which applies to languages $L \subseteq A^\infty$.
Let $\varphi : A^* \to M$ be a homomorphism to a finite monoid $M$.
First, we define a relation $\sim_\varphi$ as follows. 
If $u\in A^*$ is a finite word, then we write $u \sim_\varphi v$
if $v$ is finite and $\phi(u) = \phi(v)$. 
If $u\in A^\oo$ is an infinite word, then we write $u \sim_\varphi v$
if $v$ is infinite and if there are factorizations 
$u = u_1u_2 \cdots $ and $v= v_1v_2\cdots$ into finite nonempty words 
such that $\phi(u_i) = \phi(v_i)$ for all $i\geq 1$.
It is easy to see that $\sim_\phi$ is not transitive on infinite words, in general. Therefore, we consider its transitive closure $\approx_\varphi$.
If $u,v\in A^*$, then we have 
$$u \sim_\varphi v \iff u \approx_\varphi v \iff \phi(u) = \phi(v).$$
If $\alp,\bet\in A^\oo$, then we have $\alp \approx_\varphi \beta$
\IFF there is sequence of infinite words $\alp_0, \ldots \alp_k$ such that 
$$\alp = \alp_0 \sim_\varphi \cdots \sim_\varphi \alp_k = \bet.$$
We say that $L\subseteq A^\infty$ is \emph{recognizable} by $M$ if there exists a homomorphism $\varphi : A^* \to M$ such that $u \in L$ and $u \sim_\varphi v$ implies $v \in L$. We also say that $M$ or  $\phi$ recognizes $L$
in this case. 

The connection to the classical notation is as follows. A regular language $L\subseteq A^\infty$ is recognizable (in our sense) by $\phi$ \IFF the syntactic monoids of $L\cap  A^*$ and $L\cap A^\oo$ are divisors of $M$
(in the classical sense). 

Every variety $\bfV$ defines a family of regular languages $\bfV(A^\infty)$ as follows: we let $L \in \bfV(A^\infty)$ if there exists a monoid $M\in \bfV$ which recognizes $L$. 
Further, we define $\bfV(A^*) = \set{L\subseteq A^*}{L \in \bfV(A^\infty)}$ and $\bfV(A^\omega) = \set{L\subseteq A^\omega}{L \in \bfV(A^\infty)}.$
A variety of finite groups is a variety of finite monoids which contains only groups.
Throughout $\bfH$ denotes a variety of finite groups. Special cases are the \vars 
\begin{itemize}
	\item $\bfI$: %it consists in 
	the trivial group $\os 1$, only. 
	\item $\bfAb$: %it contains 
	all finite Abelian groups. 
	\item $\bfSol$: %it contains 
	all finite solvable groups. 
	\item $\bfSol_q$: %it contains 
	all finite solvable groups where the order is divisible by some power of $q$.
	\item $\bfG$: %the \var of 
	all finite groups.
\end{itemize}
 According to standard notation $\ov \bfH$ denotes the variety of finite
monoids where all subgroups belong to $\bfH$. 
It is not completely obvious, but a classical fact \cite{eil76}, that $\ov \bfH$ is indeed a variety. 
In fact, it is the maximal variety $\bfV$ such that $\bfV \cap \bfG = \bfH$. 

Clearly, $\ov \bfG$ is the class of all finite monoids. 
The most prominent subclass is $\ov \bfI$: it is the variety 
of aperiodic monoids $\AP$. 
The class $\AP(A^\infty) =\ov\bfI(A^\infty)$ admits various other characterizations as subsets of $A^\infty$. For example, it is the class of star-free languages $\SF(A^\infty)$, it is the class of first-order definable languages, and it is the class of definable languages in linear temporal logic over finite or infinite words: $\LTL(A^\infty)$ . 
%For a survey we refer to \cite{dg08SIWT}.

{\bf Local divisors.}
Let $M$ be a finite monoid and $c\in M$. 
Consider the set $cM \cap Mc$ with a new multiplication $\circ$ which is defined as follows: 
$$mc \circ cn = mcn.$$
A straightforward calculation shows that 
$cM \cap Mc$ becomes a monoid with this operation where the
neutral element of $M_c$ is $c$. Thus, the structure $M_c = (cM \cap Mc, \circ, c)$ defines a monoid. We say that $M_c$ is the \emph{local divisor} of $M$ at $c$. If $c$ is a unit, then $M_c$ is isomorphic to $M$. If $c= c^2$, then $M_c$ is the standard ``local monoid'' at the idempotent $c$. 

The important fact is that $M_c$ is always a divisor of $M$ and that 
$\abs{M_c}<\abs M$ as soon as $c$ is not a unit of $M$. 
Indeed, the mapping $\lam_c : \set{x\in M}{cx \in Mc} \to M_c$ given by $\lam_c(x) = cx$ is a surjective \homo. Moreover, 
if $c$ is not a unit, then $1 \notin cM \cap Mc$, hence 
$\abs{M_c}<\abs M$.
Thus, if $M$ belongs to some variety $\bfV$, then $M_c$ belongs to the same variety. If $M$ is not a group, then we find some nonunit $c\in M$ and
the local divisor $M_c$ is smaller than $M$. This makes the construction useful for induction. For a survey on the local divisor technique we refer to \cite{DiekertK2015tcs}. 

{\bf Rees extensions.}
Let $N, L$ be monoids and $\rho: N \to L$ be any mapping. 
The \emph{Rees extension} $\Rees{N,L,\rho}$ is a classical construction
for monoids \cite{pin86,rs09qtheory}, frequently described in terms of matrices. 
Here, we use an equivalent definition as in \cite{DiekertKW12tcs}.  
As a set we define 
$$\Rees{N,L,\rho} = N \cup N\times L \times N.$$ 
The multiplication $\cdot$ on $\Rees{N,L,\rho}$ is given by 
\begin{align*}
n\cdot n' &= nn' &&\text{for } n,n'\in N,\\
n\cdot (n_1,m,n_2) \cdot n' &= (nn_1, m , n_2n') &&\text{for } n,n',n_1,n_2\in N, m\in L,\\
(n_1,m,n_2)\cdot (n_1',m',n_2') &= (n_1,m\rho(n_2n_1')m',n_2') &&\text{for } n_1,n_1',n_2,n_2'\in N, m,m'\in L.
\end{align*}
The neutral element of $\Rees{N,L,\rho}$ is $1 \in N$ and %the inclusion 
$N \sse \Rees{N,L,\rho}$ is an embedding of monoids. In general, $L$ is not a divisor of $\Rees{N,L,\rho}$. The following property holds. 
%Example: Rees(1,{1,a,0}) hat keine Idempotenten ausser 1 und 0, kann also {1,a,0} mit a^2 = a nicht enthalten
\begin{lemma}\label{lem:Reesdivisors}
	Let $N \preceq N'$ and $L \preceq L'$. Given $\rho : N \to L$, there exists a mapping $\rho' : N' \to L'$ such that $\Rees{N,L,\rho}$ is a divisor of $\Rees{N',L',\rho'}$.
\end{lemma}
\begin{proof}
First, assume that $N$ (resp.~$L$) is submonoid in $N'$ (resp.~$L'$). Let $\rho' : N' \to L'$ be any function such that $\rho'|_N = \rho$. 
The mapping $\pi : \Rees{N,L,\rho} \to \Rees{N',L',\rho'}$ given by $\pi(n) = n$ and $\pi(n_1,{\ell},n_2) = (n_1,{\ell},n_2)$ is an injective homomorphism.

Second, let $\varphi : N' \to N$ and $\psi : L' \to L$ be surjective homomorphisms. Let $\rho' : N' \to L'$ be a function such that $\rho'(n) \in \psi^{-1}(\rho(\varphi(n)))$. Let $\pi : \Rees{N',L',\rho'} \to \Rees{N,L,\rho}$ be the mapping defined by $\pi(n) = \varphi(n)$ and $\pi(n_1,{\ell},n_2) = (\varphi(n_1),\psi({\ell}),\varphi(n_2))$. It is clear that $\pi$ is surjective. It is a homomorphism since 
\begin{align*}
\pi((n_1,{\ell},n_2)\cdot (n_1',{\ell}',n_2')) = \pi(n_1,{\ell} \rho'(n_2n_1') {\ell}', n_2') 
&= (\varphi(n_1),\psi({\ell})\underbrace{\psi(\rho'(n_2n_1'))}_{=\rho(\varphi(n_2n_1'))}\psi({\ell}'),\varphi(n_2')) \\
= (\varphi(n_1),\psi({\ell}),\varphi(n_2))\cdot (\varphi(n_1'),\psi({\ell}'),\varphi(n_2')) %\\
&= \pi(n_1,{\ell},n_2)\cdot \pi(n_1',{\ell}',n_2'). %\qedhere
\end{align*}
The result follows because $\preceq$ is transitive. 
\end{proof}

We are mainly interested in the case where $N$ and $L$ are proper divisors of a given finite monoid $M$. This leads to the notion of local Rees monoids.  More precisely, let $M$ be a finite monoid, $N$ by a proper submonoid of $M$ and $M_c$ be a local divisor of $M$ at $c$ where $c$ is not a unit. 
The \emph{local Rees product} $\LocRees{N,M_c}$ is defined as the Rees extension $\Rees{N,M_c,\rho_c}$ where $\rho_c$ denotes the mapping
 $\rho_c:N\to M_c; x \mapsto cxc$.

For a variety $\bfV$ we define $\Rees{\bfV}$ to be the least variety which contains $\bfV$ and is closed under taking Rees products and $\LocRees{\bfV}$ to be the least variety which contains $\bfV$ and is closed under local Rees products. 

\subsection{\schuetz{}'s $\SD$ classes}\label{sec:sd}
%Let $A$ be a finite alphabet.
 Schützenberger gave a language theoretical characterization of the class of star-free languages  $\SF(A^*)$ avoiding complementation, but
allowing the star-operation to prefix codes of bounded synchronization delay \cite{Schutzenberger1975d}.

A language $K \sse A^+$ is called \emph{prefix code} if it is \emph{prefix-free}. That is:  $u, uv \in K$
implies $u=uv$. A prefix-free language $K$ is a code
since every word $u \in K^*$
admits a unique factorization $u = u_1 \cdots u_k$ with $k \geq 0$
and $u_i \in K$. 
Note that the empty set $\es$ is considered to be a prefix code. %(The value $k$ is unique because $1 \notin K$.)
More generally, if $L\sse A^+$ is any subset, then 
$K=L\sm LA^+$ is a prefix code. 
A prefix code $K$ has \emph{bounded synchronization delay} if for some
$d \in \N$ and for all $u,v, w \in A^*$ we have:
$ %\begin{equation*}
  \text{if \,} uvw\in K^* \,\text{ and }\, v \in K^d 
  \text{, \,then }\, uv \in K^*
$. %\end{equation*}
Note that the condition implies that for all $uvw\in K^*$
with $v \in K^d$, we have $w\in K^*$, too. 
If $d$ is given explicitly, $K$ has said to have synchronization
delay $d$.  Every subset $B \sse A$ (including the empty set) yields a prefix code
with synchronization delay $0$.  If we have $c\in A \sm B$, then 
$B^*c$ is a prefix code with synchronization delay $1$.
If $K$ is any prefix code with (or without) bounded synchronization delay, then $K^m$ is a prefix code for all $m \in \N$, but for $m \geq 2$ it is never of bounded synchronization delay.

Let $G$ be a finite group. By $\SDH{G}(A^\infty)$ we denote the set of regular languages which is inductively defined as follows. 
\begin{enumerate}%[ \,1.]
\item We let $\es \in \SDH{G}(A^\infty)$ and $\smallset{a} \in
  \SDH{G}(A^\infty)$ for all letters $a \in A$.
\item If $L, K \in\SDH{G}(A^\infty)$, then %the languages 
$L \cup K$ and $(L\cap A^*)\cdot K$ are both in $\SDH{G}(A^\infty)$.
\item Let $K\sse A^+$ be a prefix code of bounded synchronization delay and $\gamma_K : K \to G$ be any mapping of $K$ to the group $G$ such that $\gamma_K^{-1}(g) \in \SDH{G}(A^\infty)$ for all $g\in G$.  %Then we have $\gamma^{-1}(1) = \gamma^{-1}(1)^*\sse A^*$. 
We let  $\gamma^{-1}(1) \in \SDH{G}(A^\infty)$ 
and $\gamma^{-1}(1)^\oo\in \SDH{G}(A^\infty)$, where $\gam: K^* \to G$ denotes the canonical extension of $\gam_K$ to a \homo from the free submonoid $K^* \sse A^*$ to $G$.
\end{enumerate}

\noindent We also define 
\begin{align*}\SDH{G}(A^*)=\set{L\sse A^*}{L\in \SDH{G}(A^\infty)} &\text{\quad and \quad} \SDH{G}(A^\oo)=\set{L\sse A^\oo}{L\in\SDH{G}(A^\infty)}.
\end{align*}
Note that for every \homo $\gam: A^*\to G$ we have 
$\gamma^{-1}(1) \in \SDH{G}(A^*)$ and $\gamma^{-1}(1)^\oo \in \SDH{G}(A^\oo)$.
This follows because first, $A$ is a prefix code of bounded synchronization delay and second, all finite subsets of $A$ are in $\SDH{G}(A^*)$. 

Unlike the case of star-free sets, the inductive
definition of $\SDH G(A^\infty)$ does not use any complementation. 
By induction: for $L \sse A^\infty$ we have 
$L \in \SDH G(A^\infty)$ \IFF we can write $L = L_1 \cup L_2$ with $L_1 \in \SDH G(A^*)$  and $L_2 \in \SDH G(A^\oo)$.
In the special case where $G=\os 1$ is the trivial group,  we also simply write $\SD$ instead of $\SDH {\os 1}$. 
In this case 
the third condition can be rephrased in simpler terms as follows. 
\begin{itemize}%[ \,1.]
\item If $K\in \SDH{}(A^*)$ is a prefix code of bounded synchronization delay, then $K^*\in \SDH{}(A^*)$ and 
$K^\oo \in \SDH{}(A^\oo)$.
\end{itemize}

In \cite{Schutzenberger1974b} Sch\"utzenberger showed  (using a different notation)
$\SDH{\bfH}(A^*) \sse \ov {\bfH}(A^*)$, but the converse only for $\bfH \subseteq \bfAb$, see \prref{prop:schuetzSDG} for the first inclusion.
Our aim is to show $\ov {\bfH}(A^\infty) \sse \SDH{\bfH}(A^*)$ for all $\bfH$, cf.~\prref{thm:characterizationSD}.  
We begin with a technical lemma. 

\begin{lemma}\label{lem:generalstar}
Let	$K\sse A^+$ be a prefix code of bounded synchronization delay and let $\gamma : K^* \to G$ be a \homo such that $\gamma^{-1}(g)\cap  K\in \SDH{G}(A^*)$ for all $g\in G$, then we have $\gamma^{-1}(g) \in \SDH{G}(A^*)$ for all $g\in G$.
\end{lemma}
\begin{proof}
For each $w \in K^*$ we construct a language $L(w) \in \SDH{G}(A^*)$ such that
\begin{itemize}
	\item $w \in L(w) \sse \gamma^{-1}(\gamma(w))$,
	\item $\abs{\set{L(w)}{w \in K^*}} < \infty$.
\end{itemize}

Consider $w= u_1 \cdots u_k \in \gamma^{-1}(g)$ with 
$u_i\in K$.   Define 
$P(w)= \set{\gam(u_1\cdots u_i)}{1\leq i \leq k} \sse G$
to be the set of prefixes of $w$ in $G$. We perform an induction on $\abs{P(w)}$. The case $\abs{P(w)} = 0$ implies $g=1$. Hence, we let $L(w) = \gamma^{-1}(1)$; and we have $\gamma^{-1}(1)\in \SDH{G}(A^*)$ by definition.
Hence, we may assume $\abs{P(w)} \geq 1$. Let $g_1= \gam(u_1)$ and 
choose $i$ maximal such that $g_1= \gam(u_1\cdots u_i)$. 
Then we have $u_1 \cdots u_i \in (K\cap \gamma^{-1}(g_1)) \cdot \gamma^{-1}(1)$. 
Define $w' = u_{i+1} \cdots u_k$. By maximality of $i$ we have 
$\abs{\set{\gam(u_1\cdots u_j)}{i < j \leq k}}< \abs{P(w)}$ because $P(w') = g_1^{-1}\cdot\set{\gam(u_1\cdots u_j)}{i < j \leq k}$. %, that is $\abs{P(w')} < \abs{P(w)}$. 
By induction there exists $L(w')$ (and only a finite number of them); and we let $L(w) = (K\cap \gamma^{-1}(g_1)) \cdot \gamma^{-1}(1) \cdot L(w')$.  
The result follows because we can write $\gamma^{-1}(g) = \bigcup\set{L(w)}{w \in \gamma^{-1}(g)}$ and this is a finite union.
\end{proof}

Clearly, we have for all $G$: if $K\in \SDH{G}(A^*)$ is a prefix code of bounded synchronization delay, then $K^*$ and 
$K^\oo$ are both in $\SDH{G}(A^\infty)$. 
As a special case, using the prefix code $K = \emptyset$, it holds $K^* = \os{1} \in \SDH{G}(A^\infty)$. 
More generally, every finite language is in $\SDH{G}(A^\infty)$. 
Note also that for $G' \leq G$ we have 
$\SDH{G'}(A^\infty) \sse \SDH{G}(A^\infty)$. In particular, 
$\bigcup{\set{\SDH{G_i}(A^\infty)}{i \in I}}\sse \SDH{\prod_{i\in I}G_i}(A^\infty)$ for every finite index set $I$. 
This inclusion holds for every divisor of $G$ as observed by the next lemma.
\begin{lemma}\label{lem:toy}
	$\SDH{H}(A^\infty) \sse \SDH{G}(A^\infty)$ holds for $H \preceq G$. 
\end{lemma}
\begin{proof}
	Inductively, it suffices to prove that $\gamma^{-1}(1), \gamma^{-1}(1)^\omega \in \SDH{G}(A^\infty)$ for a prefix code $K\sse A^+$ of bounded synchronization delay and $\gamma : K^* \to H$ a \homo of the free monoid $K^*$ 
	to the group $H$ such that $K\cap \gamma^{-1}(h) \in \SDH{G}(A^\infty)$ for all $h\in H$. Without loss of generality we may assume that there exists a surjective homomorphism $\pi : G \to H$. 
	Let $g_h \in G$ be elements such that $\pi(g_h) = h$.
	Let $\psi : K^* \to G$ be the homomorphism such that $\psi(u) = g_{\gamma(u)}$ for $u\in K$. By definition it holds $\gamma = \pi \circ \psi$. 
	Now $K \cap \psi^{-1}(g_h) = K \cap \gamma^{-1}(h) \in \SDH{G}(A^\infty)$ and $K \cap \psi^{-1}(g) = \emptyset$ if $g\neq g_h$ for all $h\in H$. Thus, $\psi^{-1}(1), \psi^{-1}(1)^\omega \in \SDH{G}(A^\infty)$ and by \reflem{lem:generalstar} we also have $\psi^{-1}(g) \in \SDH{G}(A^\infty)$ for all $g\in G$.
	Note that
	\begin{align*}
	\gamma^{-1}(1) &= \bigcup_{\pi(g)=1} \psi^{-1}(g) \quad\text{ and}\\
	\gamma^{-1}(1)^\omega &= \bigcup_{\pi(g)=1} \psi^{-1}(g) \psi^{-1}(1)^\omega
	\end{align*}
	which proves that $\gamma^{-1}(1), \gamma^{-1}(1)^\omega \in \SDH{G}(A^\infty)$.
\end{proof}

We will formulate our results on the language classes $\SDH{G}(A^\infty)$ to obtain finer results, however our main result then is formulated with the language class $$\SDG{\bfH}(A^\infty) = \bigcup{\set{\SDH{G}(A^\infty)}{G \in \bfH}}.$$

The main result is the following equality between $\SDG{\bfH}, \varietyHline$ and $\LocRees{\varietyH}$.
\begin{theorem}\label{thm:characterizationSD}
	Let $L\subseteq A^\infty$ be a regular language and $\varietyH$ a variety of finite groups. Then the following properties are equivalent:
	\begin{enumerate}
		\item $L \in \SDH{\bfH}(A^\infty)$.\label{item:thmmain:a}
		\item $L \in \varietyHline(A^\infty)$.\label{item:thmmain:b}
		\item $L \in \LocRees{\varietyH}(A^\infty)$.\label{item:thmmain:c1}
		%\item $L \in \Rees{\ov \varietyH}(A^\infty)$.\label{item:thmmain:c2}
	\end{enumerate}
\end{theorem}
\begin{corollary}
	$\SDG{\bfH}(A^\infty)$ is closed under complementation and intersection for every variety $\bfH$ of finite groups.
\end{corollary}
\begin{proof}
	By \refthm{thm:characterizationSD} we have $\SDH{\bfH}(A^\infty) = \varietyHline(A^\infty)$ and $\varietyHline(A^\infty)$ is closed under complementation and intersection.
\end{proof}
The proof of \refthm{thm:characterizationSD} 
covers the next three sections.

\section{Closure properties of $\SDG{\bfH}$}\label{sec:sdclosure}
In this section we prove the direction \ref{item:thmmain:a} $\implies$ \ref{item:thmmain:b} of \refthm{thm:characterizationSD}. 
Therefore one has to study the closure properties under the operations given in the definition of $\SDG{\bfH}(A^\infty)$, that is, one has to show that those operations do not introduce new groups.
%there appear no new groups using the closure operations.  

The following proposition of Sch\"utzenberger shows that the operation $\gamma^{-1}(1)$ does not introduce new groups. 
\begin{proposition}[\cite{Schutzenberger1974b}]\label{prop:schuetzSDG}
	Let $K\sse A^+$ be a prefix code of bounded synchronization delay and $\gamma_K : K \to G$ be a mapping such that $K_g= \oi \gamma_K(g)$ are  regular languages for $g \in G$. Let $\gamma : K^* \to G$ be the \homo from the free submonoid $K^*$ of $A^*$ 
	to the group $G$ such that $\gamma|_K = \gamma_K$. 
	View $\gamma^{-1}(1)$ as a subset of $A^*$. 
	Then, subgroups in the syntactic monoid of the language $\gamma^{-1}(1)$ are either divisors of  $G$ or  of the direct product $\prod_{g\in G} \Synt(K_g)$.
\end{proposition}
We will prove the same for $\gamma^{-1}(1)^\omega$, relying on \refprop{prop:schuetzSDG} as a blackbox result. The concept used for transfering the properties to infinite words are Birget-Rhodes expansions \cite{br84,br89}.
The Birget-Rhodes expansion of a monoid $M$ is the monoid $\BRExp{M} = \set{(X,m)}{1,m \in X\subseteq M}.$ 
The multiplication on $\BRExp{M}$ is given as a ``semi-direct product'': 
%\begin{align*}
$(X,m)\cdot (Y,n) = (X\cup m\cdot Y, m\cdot n).$
%\end{align*}
Note that $M$ is isomorphic to the submonoid $\set{(M,m)}{m\in M}$ of $\BRExp{M}$, that is, $M$ is a divisor of $\BRExp{M}$.
Moreover, the following lemma shows that the Birget-Rhodes expansion has the same groups as $M$.
\begin{lemma}\label{lem:BRgroups}
	Every group contained in $\BRExp{M}$ is isomorphic to some group in $M$.
\end{lemma}
\begin{proof}
	Let $G\subseteq \BRExp{M}$ be a group contained in $\BRExp{M}$ and let $(X, e) \in G$ be the unit in $G$. For every element $(Y,m) \in G$ we have $(X,e)(Y,m) = (X\cup eY, em) = (Y,m)$ and thus $X\subseteq Y$. Furthermore, $(Y,m)^{\abs{G}} = (Y\cup \ldots, e) = (X,e)$ and we conclude $X=Y$. Thus, $(X,m) \mapsto m$ is an injective embedding of $G$ in $M$.
\end{proof}

The idea behind the Birget-Rhodes expansion is that it stores the seen prefixes in a set. More formally, the following lemma holds.
\begin{lemma}\label{lem:brexpprefix}
	Let $\varphi : A^* \to M$ be a homomorphism and $\psi : A^* \to \BRExp{M}$ be the homomorphism given by $\psi(a) = (\os{1,\varphi(a)},\varphi(a))$. Let $u\in A^*$ and $\psi(u) = (X,\varphi(u))$. For every $m\in X$ there exists a prefix $v$ of $u$ such that $\varphi(v) = m$.
\end{lemma}
\begin{proof}
	We will prove this inductively. The statement is true if $u$ is the empty word. 
	Thus, consider $u = va$ for some letter $a\in A$. Let $\psi(v) = (Y,\varphi(v))$, then \[\psi(u) = \psi(v) \cdot (\os{1,\varphi(a)},\varphi(a)) = (Y\cup \os{\varphi(v),\varphi(v)\varphi(a)}, \varphi(u)).\] 
	Inductively, we obtain prefixes of $v$, and therefore also prefixes of $u$, for all elements of $Y$. The only (potentially) new element in $X$ is $\varphi(u)$. This proves the claim.
\end{proof}
A special kind of $\omega$-regular languages are \emph{arrow languages}. 
Let $L \subseteq A^*$ be a language. We define 
%\begin{equation*}
$\overrightarrow{L} = \set{\alpha \in A^\omega}{\text{infinitely many prefixes of } \alpha \text{ are in } L}$
%\end{equation*}
to be the arrow language of $L$. The set of arrow languages is exactly the set of deterministic languages \cite{tho90handbook}. 
The Birget-Rhodes expansion can be used to obtain a recognizing monoid for $\overrightarrow{L}$, given a monoid for $L$.
\begin{proposition}\label{prop:BRarrow}
	Let $L\subseteq A^*$ be some regular language and $\varphi : A^* \to M$ be a homomorphism which recognizes $L$, then $\overrightarrow{L}$ is recognized by $\BRExp{M}$.
\end{proposition}
\begin{proof}
	Let $\psi : A^* \to \BRExp{M}$ be the homomorphism given by $\psi(a) = (\os{1,\varphi(a)},\varphi(a))$. 
	Let $\alpha \in \overrightarrow{L}$ and $\alpha \sim_\psi \beta$. We show that $\beta \in \overrightarrow{L}$. 
	Let $\alpha = u_1u_2\cdots$ and $\beta = v_1v_2\cdots$ be factorizations such that $\psi(u_i) = \psi(v_i)$. 
	Since $\alpha \in \overrightarrow{L}$, we may assume that for every $i$ there exists a decomposition $u_i = u_i' u_i''$ such that $u_1\cdots u_{i-1}u_i' \in L$. 
	By $\psi(u_i) = \psi(v_i)$ and \reflem{lem:brexpprefix}, there exists a decomposition $v_i = v_i'v_i''$ such that $\varphi(u_i') = \varphi(v_i')$. Thus, $u_1\cdots u_{i-1}u_i' \sim_\varphi v_1\cdots v_{i-1}v_i'$ and therefore $v_1\cdots v_{i-1}v_i' \in L$. This implies $\beta \in \overrightarrow{L}$. 
\end{proof}

We are now ready to show the main result of this section, that is, every language in $\SDH{G}(A^\infty)$ has only groups which are divisors of direct products of $G$. In particular, this implies $\SDG{\bfH}(A^\infty) \subseteq \varietyHline(A^\infty)$.
\begin{proposition}
	If $L \in \SDH{G}(A^\infty)$, then all subgroups in $\Synt(L)$ are a divisor of a direct product of copies of $G$.
\end{proposition}
\begin{proof}
	We will prove this inductively on the definition of $\SDH{G}(A^\infty)$.
	The cases $\es \in \SDH{G}(A^\infty)$ and $\smallset{a} \in
	\SDH{G}(A^\infty)$ for all letters $a \in A$ are straightforward, as they are recognized by aperiodic monoids. 
	Let $L, K$ be languages, such that their syntactic monoids contain only groups which are divisors of a direct product of $G$. 
	The language $L \cup K$ is recognized by the direct product of their syntactic monoids which implies the statement. 
	$(L\cap A^*)\cdot K$ is recognized by the Sch\"utzenberger product of their syntactic homomorphisms \cite[Proposition 11.7.10]{dr95}. The Sch\"utzenberger product does not introduce new groups \cite{sch65sf}\footnote{A proof of these two citations also can be found in the appendix.}.
		
	Let $K\sse A^+$ be a prefix code of bounded synchronization delay and $\gamma : K^* \to G$ be a \homo of the free monoid $K^*$ 
	to the group $G$ such that for all $g\in G$
	every subgroup of $\Synt(K\cap \gamma^{-1}(g))$ is a divisor of a direct product of copies of $G$.
	\refprop{prop:schuetzSDG} implies that every subgroup of $\Synt(\gamma^{-1}(1))$ is a divisor of a direct product of copies of $G$.  
	Note that $\gamma^{-1}(1)^\oo = \overrightarrow{\gamma^{-1}(1)}$ and therefore \refprop{prop:BRarrow} and \reflem{lem:BRgroups} imply that every subgroup of $\Synt(\gamma^{-1}(1)^\omega)$ is a divisor of a direct product of copies of $G$.
\end{proof}

\section{The inclusion $\ov{\bfH}(A^\infty)\subseteq \SDG{\bfH}(A^\infty)$}\label{sec:ldproof}
In this section we prove the direction \ref{item:thmmain:b} $\implies$ \ref{item:thmmain:a}. 
We prove that if every subgroup of $M$ is a divisor of $G$, then every language recognized by $M$ is contained in $\SDH{G}(A^\infty)$. This result is again finer than just the inequality $\ov{\bfH}(A^\infty)\subseteq \SDG{\bfH}(A^\infty)$. The proof works by induction on $\abs M$ and on the alphabet and decomposes every $\approx_\varphi$-class into several sets in $\SDH{G}(A^\infty)$.
\begin{proposition}
	Let $L \subseteq A^\infty$ be recognized by $\varphi : A^* \to M$ and let $G$ be a group such that every subgroup of $M$ is a divisor of $G$, then $L \in \SDH{G}(A^\infty)$. Moreover, $L$ can be written as finite union
	\begin{align*}
	L = L_0 \cup \bigcup_{i=1}^m L_i\cdot \gamma_i^{-1}(1)^\omega
	\end{align*}
	for $L_i \in \SDH{G}(A^*)$ and $\gamma_i : K_i^* \to G$ for prefix codes $K_i \in \SDH{G}(A^*)$ of bounded synchronization delay with $\gamma_i^{-1}(g) \cap K_i \in \SDH{G}(A^*)$ for all $g\in G$. All products in the expressions of $L_i$ are unambiguous.
\end{proposition}
\begin{proof}
	Let $\dbracket{w}_\varphi = \set{v\in A^\infty}{w \approx_\varphi v}$ be the equivalence class of $w$. Since $L$ is recognized by $\varphi$, it holds $L = \cup_{w\in L} \dbracket{w}_\varphi$.
	Our goal is to construct languages $L(w) \in \SDH{G}(A^\infty)$ such that 
	\begin{itemize}
		\item $w \in L(w) \subseteq \dbracket{w}_\varphi$.
		\item the number of such languages is bounded by some function in $\abs{A}$ and $\abs{M}$.
		\item every word in $L(w)$ starts with the same letter.
	\end{itemize} 
	In particular, we want to saturate $\dbracket{w}_\varphi$ by sets in $\SDH{G}(A^\infty)$.
	The construction of the set $L(w)$ is by induction on $(\abs{M},\abs{A})$ with lexicographic order.
	
	If $w=1$, then we set $L(w) = \os{1}$. This concludes the induction base $\abs A = 0$.
	Let us consider the case that $\varphi(A^*)$ is a group, that is, a divisor of $G$. 
	Consider the prefix code $K = A$ of synchronization delay $1$ and the homomorphism $\gamma = \varphi$. 
	Note that since $\os{a} \in \SDH{G}(A^\infty)$ and $\SDH{G}(A^\infty)$ is closed under union, every subset of $K$ is in $\SDH{G}(A^\infty)$. In particular, $K \cap \gamma^{-1}(g) \in \SDH{G}(A^\infty)$ for all $g\in \varphi(A^*)$. This shows $\gamma^{-1}(g) = \varphi^{-1}(g) \in \SDH{G}(A^*)$ for all $g\in \varphi(A^*)$ by \reflem{lem:generalstar} and \reflem{lem:toy}. 
	If $w = av \in aA^*$ for some $a \in A$, then set $L(w) = a\varphi^{-1}(\varphi(v))$. It is clear that $w \in L(w) \subseteq \dbracket{w}_\varphi$ and $L(w) \in \SDH{G}(A^\infty)$ by the above.
	If $w \in aA^\omega$, then we obtain $w \in a\varphi^{-1}(m)\varphi^{-1}(1)^\omega$ for some $m \in M$ by Ramsey's theorem. The idempotent in this decomposition must be $1$ since $\varphi(A^*)$ is a group. 
	Thus, we may set $L(w) = a\varphi^{-1}(m)\varphi^{-1}(1)^\omega$. Note that by the definition of $\sim_\varphi$, the inclusion $L(w) \subseteq \dbracket{w}_\varphi$ holds. 
	In particular, these cases include the induction base $\abs M = 1$.

	In the following we assume that $\varphi(A^*)$ is not a group and therefore there exists a letter $c \in A$ such that $\varphi(c)$ is not a unit. 
	Fix this letter $c\in A$ and set $B = A\setminus \os{c}$. If $w\in B^\infty$, the set $L(w)$ exists by induction. 
	Let $w = uv$ with $u \in B^*$ and $v \in cA^\infty$. By induction we obtain $L(u) \in \SDH{G}(B^\infty) \subseteq \SDH{G}(A^\infty)$ and it remains to show $L(v) \in \SDH{G}(A^\infty)$. Note that the product $L(w) = L(u)\cdot L(v)$ is unambiguous. 
	From now on we may assume $w \in cA^\infty$.  
	Let us first consider the case $w = uv$ with $u \in c(B^*c)^*$ and $v \in B^\infty$, i.e., there are only finitely many occurences of the letter $c$ in $w$. By induction, there exists $L(v) \in \SDH{G}(B^\infty) \subseteq \SDH{G}(A^\infty)$ and by setting $L(w) = L(u)\cdot L(v)$ it remains to construct $L(u)$.
	
	Consider the alphabet $T = \varphi(B^*) = \set{\varphi(u)}{u\in B^*}$. Let $M_c$ be the local divisor of $M$ at $\varphi(c)$. 
	Since $M_c$ is a divisor of $M$, every subgroup of $M_c$ is a divisor of $G$. 
	Consider the homomorphism $\psi : T^* \to M_c$ given by $\psi(\varphi(u)) = \varphi(cuc)$ and the substitution $\sigma : (B^*c)^\infty \to T^\infty$ with $\sigma(u_1c u_2c\ldots) = \varphi(u_1)\varphi(u_2)\cdots$.
	Note that 
	\begin{align*}
	\psi(\sigma(u_1cu_2c\ldots u_nc)) &= \psi(\varphi(u_1)\varphi(u_2) \cdots \varphi(u_n))
	= \varphi(cu_1c)\circ \varphi(cu_2c) \circ \cdots \circ \varphi(cu_nc) \\
	&= \varphi(cu_1cu_2c\ldots cu_nc)
	\end{align*}
	and thus $\varphi^{-1}(m) \cap c(B^*c)^* = c\sigma^{-1}(\psi^{-1}(m))$. 
	By induction on the monoid size, since $\abs{M_c} < \abs{M}$, there exists a language $L(\sigma(u')) \in \SDH{G}(T^\infty)$ for all $u' \in (B^*c)^*$. 
	We show $\sigma^{-1}(K) \in \SDH{G}(A^\infty)$ for all $K \in \SDH{G}(T^\infty)$ inductively on the definition of $\SDH{G}$. 
	Then we can set $L(u) = c\sigma^{-1}(L(\sigma(u')))$ for $u = cu'$ and have completed the case of finitely many $c$'s.

	For $K=\emptyset$, we obtain $\sigma^{-1}(K) = \emptyset \in \SDH{G}(A^\infty)$. Furthermore, \[\sigma^{-1}(t) = \bigcup_{v\in B^*,  t=\varphi(v)} L(v)c \in \SDH{G}(A^\infty).\] 
	Let $L,K \in \SDH{G}(T^\infty)$. 
	A basic result from set theory yields 
	%\begin{align*}
	$\sigma^{-1}(L\cup K) = \sigma^{-1}(L) \cup \sigma^{-1}(K)$.
	%\end{align*}
	Let $\sigma(v) = w_1w_2$ for some $v\in (B^*c)^*$. 
	Since $B^*c$ is a prefix code, there exists a unique factorization $v = v_1v_2$ with $v_1,v_2 \in (B^*c)^*$ such that $\sigma(v_1) = w_1$ and $\sigma(v_2) = w_2$. 
	Thus, we conclude $\sigma^{-1}(K\cdot L) = \sigma^{-1}(K)\cdot \sigma^{-1}(L)$. 
	Let now $K \in \SDH{G}(T^\infty)$ be a prefix code of synchronization delay $d$. 
	We first show that $\sigma^{-1}(K)$ is a prefix code of bounded synchronization delay.  
	Let $u, uv \in \sigma^{-1}(K)$, then $\sigma(u), \sigma(uv) = \sigma(u)\sigma(v)\in K$ and therefore $\sigma(v)=1$. 
	This implies $v=1$ and $\sigma^{-1}(K)$ is a prefix code. 
	We prove that $\sigma^{-1}(K)$ has synchronization delay $d+1$. 
	The incrementation of the synchronization delay by one comes from the fact that $B^*c$ is not a suffix code, and thus we need another word in $B^*c$ to pose as a left marker. 
	Consider $uvw \in \sigma^{-1}(K)^*$ with $v\in \sigma^{-1}(K)^{d+1}$ and factorize $v = v_1cv_2$ with $v_2\in \sigma^{-1}(K)^d = \sigma^{-1}(K^d)$. 
	Then $\sigma(uvw) = \sigma(uv_1c)\sigma(v_2)\sigma(w)$, and by $\sigma(v_2) \in K^d$ this implies $\sigma(uv) = \sigma(uv_1c)\sigma(v_2) \in K^*$. Thus, $uv \in \sigma^{-1}(K)^*$. 
	Let $\gamma : K^* \to G$ be some homomorphism and $K_g = K\cap \gamma^{-1}(g) \in \SDH{G}(T^\infty)$ for all $g\in G$. Inductively, $\sigma^{-1}(K_g) \in \SDH{G}(A^\infty)$ and $\sigma^{-1}(K) = \bigcup \sigma^{-1}(K_g)$. Let $\gamma' : \sigma^{-1}(K)^* \to G$ be induced by $\gamma'(u) = \gamma(\sigma(u))$. By definition of $\SDH{G}(A^\infty)$ we obtain $\gamma'^{-1}(1) \in \SDH{G}(A^\infty)$. 
	However, 
	$u_1\cdots u_n \in \sigma^{-1}(\gamma^{-1}(1))$ if and only if $\gamma(\sigma(u_1\cdots u_n)) = 1$. Furthermore, note that $\gamma(\sigma(u_1\cdots u_n))  = \gamma(\sigma(u_1)) \cdots \gamma(\sigma(u_n)) = \gamma'(u_1) \cdots \gamma'(u_n) = \gamma'(u_1\cdots u_n)$. Thus, we obtain $\sigma^{-1}(\gamma^{-1}(1)) = \gamma'^{-1}(1) \in \SDH{G}(A^\infty)$ and $\sigma^{-1}(\gamma^{-1}(1)^\omega) = \gamma'^{-1}(1)^\omega \in \SDH{G}(A^\infty)$.

	The last case of the proof is that $w$ contains infinitely many $c$'s, that is, $w = cv$ with $v\in (B^*c)^\omega$. 
	By induction, we know that $\sigma(v) \in L_T\cdot \gamma_T^{-1}(1)^\omega \subseteq \dbracket{\sigma(v)}_\psi$ for some $L_T \in \SDH{G}(T^*)$ and $\gamma_T : K_T^* \to G$ for some prefix code $K_T \in \SDH{G}(T^*)$ of bounded synchronization delay with $\gamma_T^{-1}(g) \cap K_T \in \SDH{G}(T^*)$. 
	By the calculation above, there exists a $\gamma : K^* \to G$ with the usual properties such that $\gamma^{-1}(1) = \sigma^{-1}(\gamma_T^{-1}(1))$. Let $L = \sigma^{-1}(L_T)$ and set $L(w) = cL\gamma^{-1}(1)^\omega$. 
	%Since the product $L_T\cdot \gamma_T^{-1}(1)^\omega$ is unambiguous, the product in $L(w)$ is unambiguous as well. 
	It remains to show that $cL\gamma^{-1}(1)^\omega \subseteq \dbracket{w}_\varphi$. Let $cu \in cL\gamma^{-1}(1)^\omega$, then $\sigma(u) \in \dbracket{\sigma(v)}_\psi$, that is $\sigma(u) \approx_\psi \sigma(v)$. 
	Since $\approx_\psi$ is the transitive closure of $\sim_\psi$, we show that $\sigma(u) \sim_\psi \sigma(v)$ implies $cu \approx_\varphi cv$ for all $u,v \in (B*c)^\omega$ which concludes the proof. 
	Now, let $\sigma(u) = \sigma(u_1c)\sigma(u_2c)\cdots$ and $\sigma(v) = \sigma(v_1c)\sigma(v_2c)\cdots$ such that $\psi(\sigma(u_ic)) = \psi(\sigma(v_ic))$. As observed above, this implies $\varphi(cu_ic) = \varphi(cv_ic)$. Thus,
	\begin{align*}
	cu &= (cu_1c)u_2(cu_3c)u_4(c\cdots 
	\sim_\varphi (cv_1c)u_2(cv_3c)u_4(c\cdots \\
	&= cv_1(cu_2c)v_3(cu_4c)\cdots
	\sim_\varphi cv_1(cv_2c)v_3(cv_4c)\cdots \\
	&= cv.
	\end{align*}
	This implies the existence of sets $L(w) \in \SDH{G}(A^\infty)$ with $w \in L(w) \subseteq \dbracket{w}_\varphi$ in the case of infinitely many $c$'s.
\end{proof}

\section{Rees extension monoids}\label{sec:rees}
In this section we prove the direction \ref{item:thmmain:b} $\iff$ \ref{item:thmmain:c1}.
We need the fact that every group contained in $\Rees{N,M,\rho}$ is  contained in $N$ or in $M$. 
\begin{lemma}[\cite{AlmeidaK16}]\label{lem:closureRees}
	Let $G$ be a group in $\Rees{N,M,\rho}$, then there exists an embedding of $G$ into  $N$ or into  $M$.
\end{lemma}
Thus, \reflem{lem:closureRees} implies 
$\LocRees{\bfH} \subseteq \Rees{\bfH} \subseteq \Rees{\ov \bfH} \subseteq \ov{\bfH}$
for any group variety $\bfH$, which is \ref{item:thmmain:c1} $\implies$ \ref{item:thmmain:b}.
We want to prove equality, that is, every monoid which contains only groups in $\bfH$ is a divisor of an iterated Rees extension of groups in $\bfH$. However, we are able to prove a stronger statement using only local Rees extensions.
\begin{proposition}\label{prop:newlocrees}
Given $M$, we can construct a sequence 
of monoids $M_1, \ldots M_k = M$ with $k \leq 2^{\abs{M}}-1$ such that 
for each $1\leq j \leq k$ we have for 
$M_j$ one of the following:
\begin{itemize}
\item $M_j$ is a group which is a divisor of $M$. 
\item $M_j$ is a divisor of a local Rees product of some $M_i$ and a local divisor $M_\ell$ of ${M_j}$ with $i, \ell <j$.
\end{itemize}
\end{proposition}
\begin{proof}
	We proof the statement with induction on $\abs{M}$. If $M$ is a group, we set $M_1 = M$. This includes the base case $\abs{M} = 1$. 
	If $M$ is not a group, we may choose a minimal generating set of $M$. Let $c$ be a nonunit of this generating set, then there exists a proper submonoid $N$ of $M$ such that $N$ and $c$ generate $M$.
	Since $c$ is not a unit, the local divisor $M_c$ is smaller than $M$, that is, $\abs{M_c} < \abs{M}$. 
	By induction, there exist sequences $M_1',\ldots, M_{k'}' = N$ and $M_1'',\ldots, M_{k''}'' = M_c$ with $k',k'' \leq 2^{\abs{M}-1}-1$.
	%Consider the local Rees-Extension $\LocRees{N,M_c}$. 
	We show that $M$ is a homomorphic image of the local Rees product $\LocRees{N,M_c}$. Let $\varphi : \LocRees{N,M_c} \to M$ be the mapping given by $\varphi(n) = n$ for $n\in N$ and $\varphi(u,x,v)=uxv$ for $(u,x,v)\in N \times M_c \times N$. 
	Since
	\begin{align*}
	\varphi((u,x,v)(s,y,t)) &= \varphi(u,x\circ cvsc \circ y, t) = \varphi(u,xvsy,t) \\&= (uxv)(syt) = \varphi(u,x,v)\varphi(s,y,t),
	\end{align*}
	$\varphi$ is a homomorphism. Obviously, $M = N \cup NM_{c}N$ and thus $\varphi$ is surjective. 
	
	Setting $M_i = M_i'$ for $1\leq i \leq k'$, $M_{i+k'} = M_i''$ for $1 \leq i \leq k''$ and $M_{k'+k''+1} = M$ leads to such a sequence for $M$ as $M$ is a divisor of the local Rees product of $M_{k'}=N$ and $M_{k'+k''} = M_c$. Since $k'+k''+1 \leq  2\cdot (2^{\abs{M}-1}-1)+1 = 2^{\abs{M}}-1$, the bound on $k$ holds.	
\end{proof}
The inclusion $\ov{\bfH} \subseteq \LocRees{\bfH}$ is immediate from \refprop{prop:newlocrees}, which is \ref{item:thmmain:b} $\implies$ \ref{item:thmmain:c1}.
In particular, every monoid in $\ov{\bfH}$ is a divisor of an iterated Rees product of groups in $\bfH$ by \reflem{lem:Reesdivisors}.
We can draw the decomposition as a tree based on the decomposition of $M$ in submonoids and local divisors. We do not describe this formally but content ourselves to give an example.
\begin{example}\label{ex:monoid}
	Let $M$ be the monoid generated by $\os{a,b,\delta,\sigma}$ with the relations $a^2 = b^2 = ab = ba = 0$, $a\delta = a$, $\delta \sigma = \sigma \delta^2$, $\delta^3 = 1$, $\sigma^2 = 1$ and $d\delta = \delta d$, $d\sigma = \sigma d$ with $d\in \os{a,b}$. The subgroup generated by $\delta$ and $\sigma$ is the symmetric group $\mathfrak{S}_3$; it is solvable but not Abelian. The monoid $M$ is syntactic for the language $L$ which is a union of $L_a$ and $L_b$. 
The language $L_a$ is the set of all words $uav$  with $uv \in \os{\delta, \sigma}^*$ and the sign of the permutation $uv$ evaluates to $-1$.
The language $L_b$ is the set of all words $ubv$  with $uv \in \os{\delta, \sigma}^*$ and $uv$ evaluates in $\mathfrak{S}_3$ to $\delta$. The decomposition in Rees products from \refprop{prop:newlocrees} is depicted in \reffig{fig:tree}. Here $M[a,\sigma,\delta]$ denotes the submonoid generated by $\os{a,\sigma,\delta}$. In particular, this yields $M \preceq \Rees{\Rees{S_3,\mathbb Z/2\mathbb Z,\rho_1},\Rees{S_3,\os{1},\rho_2},\rho_3}$ for some $\rho_1,\rho_2,\rho_3$ by \reflem{lem:Reesdivisors}.
		\begin{figure}
			\begin{tikzpicture}[level distance=1.5cm,
			level 1/.style={sibling distance=7cm},
			level 2/.style={sibling distance=3.5cm}]
			\node {$M$}
			child {node {$M[a,\sigma,\delta]$}
				child {node {$\mathfrak{S}_3$}}
				child {node {$M[a,\sigma,\delta]_a \simeq \mathbb Z/2\mathbb Z$}}
			}
			child {node {$M_b \simeq \mathfrak{S}_3 \cup \os{0}$}
				child {node {$\mathfrak{S}_3$}}
				child {node {$(M_b)_0 \simeq \os{1}$}}
			};
			\end{tikzpicture}
			\caption{Decomposition tree of the monoid in \refex{ex:monoid}.}\label{fig:tree}
		\end{figure}
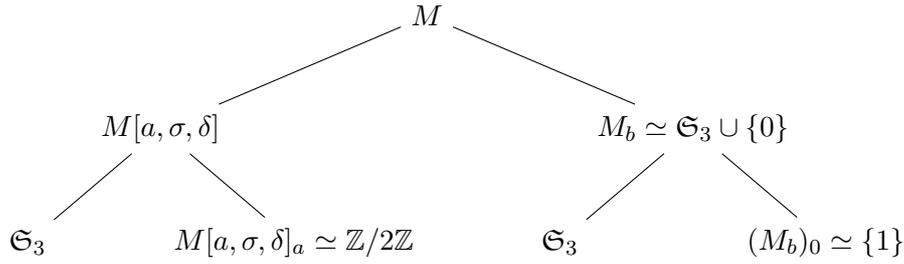
		
\end{example} 

\section{Applications}
An application of \refprop{prop:newlocrees} is the solution to an open question of Almeida and Kl{\'i}ma. Let $\bfU$ and $\bfV$ be varieties. Let $\Rees{\bfU, \bfV}$ be the variety generated by $\Rees{N,M,\rho}$ for $N\in \bfU$ and $M\in \bfV$. 
Note that in general $\Rees{\bfV} \neq \Rees{\bfV, \bfV}$. However $\Rees{\bfV}$ can be defined as the limit of this operation. Let $\bfV_i = \Rees{\bfV_{i-1},\bfV_{i-1}}$ and $\bfV_0 = \bfV$, then
\begin{equation*}
\Rees{\bfV} = \bigcup_{i\in \mathbb N} \bfV_i.
\end{equation*}
The variety $\Rees{\bfU, \bfV}$ has recently been introduced by Almeida and Kl{\'i}ma under the name of \emph{bullet operation} \cite{AlmeidaK16}. 
They defined a variety $\bfV$ to be \emph{bullet idempotent} if $\bfV = \Rees{\bfV, \bfV}$ and posed the open question whether there are varieties apart from $\ov{\bfH}$ which are bullet idempotent.
Using our decomposition above, we prove that the answer to this question is no.
\begin{theorem}
	Let $\varietyfont{V}$ be a bullet idempotent variety and let $\varietyH = \varietyfont{V} \cap \varietyG$, then $\varietyfont{V} = \varietyHline$.
\end{theorem}
\begin{proof}
	Since $\varietyHline$ is the maximal variety with $\varietyHline \cap \varietyG = \varietyH$, we have $\varietyfont{V} \subseteq \varietyHline$. Let $M \in \varietyHline$. Inductively, we may assume that every proper divisor of $M$ is in $\varietyfont{V}$. If $M$ is a group, then $M \in \varietyH$ and thus $M \in \varietyfont{V}$. 
	Thus, there exists an nonunit element $c \in M$ and a proper submonoid $N$ of $M$ such that $N$ and $c$ generate $M$. 
	By the calculation in the proof of \refprop{prop:newlocrees}, $M$ is a divisor of $\LocRees{N,M_c}$, and since $N, M_c \in \varietyfont{V}$ and $\bfV = \Rees{\bfV, \bfV}$ we obtain $M \in \varietyfont{V}$.
\end{proof}

Let $(\FO+\mathrm{MOD}_q)[<]$ be the fragment of first-order sentences which only use first-order quantifiers, modular quantifiers of modulus $q$ and the predicate $<$. Then the following theorem holds.
\begin{corollary}
	$(\FO+\mathrm{MOD}_q)[<](A^\infty) = \SDH{\bfSol_q}(A^\infty)$
\end{corollary}
\begin{proof}
	By \cite{stt95IC}, see also \cite{Straubing94} for a complete treatise,  $(\FO+\mathrm{MOD}_q)[<]$ describes the family of all regular languages such that every group in the syntactic monoid is a solvable group of cardinality dividing a power of $q$, that is the languages in $\bfSol_q$. \refthm{thm:characterizationSD} then implies the stated equality.
\end{proof}
The same language class has been described by Straubing with another operation, counting how many prefixes are in a given language, which resembles more closely the counting modulo~$q$~\cite{str79}.
\section{Summary}
\begin{table}[t!]
	\centering
	\begin{tabular}{|l|c|c|c|c|c|}
		\hline
		& $\overline{\bfI}$ & $\overline{\bfAb}$ & $\overline{\bfSol}$ & $\overline{\bfSol_q}$ & $\varietyHline$ \\ \hline
		finite words   & \cite{Schutzenberger1975d}    & \cite{Schutzenberger1974b} & \cite{str79},\textbf{new} & \cite{str79},\textbf{new}   & \textbf{new, unless} $\bfH \subseteq \bfAb$    \\ \hline
		$\omega$-words & \cite{DiekertKufleitner14tocs}     & \textbf{new}     & \textbf{new} & \textbf{new} & \textbf{new, unless} $\bfH = \bfI$ \\ \hline
	\end{tabular}
	\caption{Overview of existing and new language characterizations of $\varietyHline$.}\label{fig:charact}
\end{table}
Our main theorem \refthm{thm:characterizationSD} states $\ov{\bfH}(A^\infty) = \SDH{\bfH}(A^\infty)$. An overview over the contributions for $\varietyHline$ is given in \reffig{fig:charact}. 
%\refthm{thm:characterizationSD} directly transfer to languages over $A^*$ and languages over $A^\omega$, that is, $\ov{\bfH}(A^*) = \SDH{\bfH}(A^*)$ and $\ov{\bfH}(A^\omega) = \SDH{\bfH}(A^\omega)$. 
As a byproduct we were able to give a simple decomposition of the monoids in $\ov{\bfH}$ as local Rees products and groups in $\bfH$, using only exponentially many operations.

\bibliographystyle{plain}
%\bibliography{../../TRACES/traces}
\newcommand{\Ju}{Ju}\newcommand{\Ph}{Ph}\newcommand{\Th}{Th}\newcommand{\Ch}{Ch}\newcommand{\Yu}{Yu}\newcommand{\Zh}{Zh}\newcommand{\St}{St}\newcommand{\curlybraces}[1]{\{#1\}}

\newpage
\appendix
\section{Missing proofs}
All missing proofs can easily be deduced from the existing literature; and pointers have been given in previous sections. However, in order to keep 
the paper self-contained we reproduce them in our notation. 
We first give a proof of \refprop{prop:schuetzSDG}. The statement has been proved by Sch\"utzenberger. We give a detailed proof following\cite{Schutzenberger1974b} loosely.
We assume the reader to be familiar with basic concepts of formal language theory, such as deterministic finite automatons and remind the classic theorem that the transformation monoid of a minimal determistic finite automaton of a language is isomorphic to its syntactic monoid. 

\begin{proof}[Proof of \refprop{prop:schuetzSDG}]
	Note that $K^*$ is regular because $K = \bigcup\set{K_g}{g \in G}$ is regular. 
	Without restriction we may assume $K\neq \es$
	and we let $d$ be the synchronization delay of $K$. If $p$ denotes a state in  some deterministic finite automaton (DFA) and if $u \in A^*$ is a word, then we write 
	$p \mapsto p\cdot u$ to indicate that reading $u$ transforms  $p$ into the state $p\cdot u$. For $g\in G$ let 
	$Q_g$ be the state set of the minimal automaton for $K_g$ and $q_g$ the corresponding initial state. Let 
	$Q$ be the direct product of sets $Q_g$ with initial state
	$q_0= \prod \set{q_g}{g\in G}$. The product automaton allows to  
	assign to each language $K_g$ a subset $F_g \sse Q$ such that 
	the DFA $(Q,A,\cdot\,,q_0,F_g)$ accepts $K_g$. Since 
	$K_g \cap K_h = \es$ for $g\neq h$ we have $F_g \cap F_h = \es$ for $g\neq h$. 
	It is also clear that $\prod_{g\in G} \Synt(K_g)$ acts on $Q$. 
	
	By $F$ we denote union  $\bigcup \set{F_g}{g\in G}$.
	We  merge the subset 
	$\set{p\in  Q}{p \cdot A^* \cap F = \es}$ into a single sink state $\bot$. Since $K$ is a prefix code, there is no word $u\in A^+$ such that $p\cdot u \in F$ for any $p \in F$. 
	Thus, $p\cdot u = \bot$ for every $p\in F$ and $u\in A^+$. Moreover, 
	without restriction we may assume that every state is reachable from the initial state $q_0$ and by slight abuse of language, the new state space is still called $Q$. The image of $A^*$ in the transformation monoid 
	$Q^Q$ which is induced by $\sig_u: Q\to Q, p \mapsto p\cdot u$ defines a monoid $S$, the transition monoid of $Q$, and
	$S$  becomes a divisor of $\prod_{g\in G} \Synt(K_g)$.
	It is therefore enough to show that every subgroup in the syntactic monoid $\Synt(\gamma^{-1}(1))$ is either a divisor of $G$ or a divisor of $S$.
	For later use we denote by $\sig: A^* \to S$ the \homo 
	which maps $u$ to $\sig_u$. 
	
	Next, consider the product set $\wt Q= G \times (Q\sm F)$. We view $\wt Q$ as a state space of an automaton accepting 
	$\oi \gam(1)$ as follows. 
	$$(g,q)\cdot a = \begin{cases}
	(g,q\cdot a) & \text{if } q\cdot a \in Q\setminus F\\
	(gh, q_1) & \text{if } q\cdot a \in  F_h
	\end{cases}$$
	Note that the transition function is well-defined since, as mentioned above, $F_g \cap F_h = \es$ for $g\neq h$. 
	The construction defines a \homo 	$\mu: A^* \to {\wt Q}^{\wt Q}$. We let 
	$M = \mu(A^*)$. It is the corresponding transition monoid for $\wt  Q$. 
	Moreover, 
	letting $(1,q_1)\in \wt Q$ be the only final state, the resulting DFA accepts $\oi \gam(1)$ as a subset 
	of $A^*$. To see this observe that every word $u \in \oi \gam(1)^*$ belongs to $K^* \sse A^*$. Moreover, $u$ admits a unique factorization $u = u_1 \cdots u_k$ such that for all $i$ we have  
	$q_0 \cdot u_i \in F_{g_i}$ for $g_i= \gam(u_i)$ and 
	$1 = g_1 \cdots g_k$. 
	Since the DFA accepts $\oi \gam(1)$, it is enough to show that every  subgroup of $M$ is either a subgroup of $G$ or a divisor of $S$.

	Let $H$ be a subgroup of $M$. Then $H$ contains a unique idempotent $e\in M$ which is the neutral element in $H$. In particular, $H = eHe$. Let $\cH = \oi \mu(H)$. It is a nonempty subsemigroup of $A^*$. 
	The group $H$ does not act as a group on  $\wt Q$, because there
	might be states $(g,p)$  such that $(g,p)\neq (g,p)\cdot e$. However, 
	it acts faithfully on ${\wt Q}_e = \wt Q\cdot e$. Indeed, if $h\neq h'$ in $H$, then there are states $(g,p)\cdot h \neq (g,p)\cdot h'$.  
	Since $h = ehe$ and $h' =e h'e$, we have $(g,p)\cdot e\in {\wt Q}_e$, $(g,p)\cdot eh \neq (g,p)\cdot eh'$, and  
	$(g,p)\cdot eh, (g,p)\cdot eh'\in {\wt Q}_e$. 
	We distinguish two cases.
	
	{\bf Case 1.} There is a state $(g,p)\in {\wt Q}_e$ such that 
	there is a word $uv \in \cH$ where $p \cdot u \in F$. For $w = (uv)^{\abs H}$ we have $\mu(w) = e$ and $w$ factorizes as $w= uw'x$ such that
	$w'\in K^*$ and $q_0\cdot x = p$. It follows $xu \in K$. Letting $y = wuw'$ we have $yx= w^2\in \cH$ with $\mu(yx) = e$ and hence,  $(g,q_0)\cdot x = (g,p)$ implies $(g,p)\cdot y = (g,q_0)$. 
	
	The element $\mu(xy)$ is idempotent in $M$. 
	Indeed, calculating in $M$ we have: 
	$$(xy)^2 = x wuw' \cdot xwuw' = xw^3 uw' = x wuw' = xy.$$
	The subsemigroup $xHy$ contains the idempotent $xy$ and 
	$f \mapsto xfy$ defines a \homo of $H$ onto  the group $H'$ and its inverse is given $xfy \mapsto yxfyx = f$.
	As $H$ and $H'$ are isomorphic, we start all over with the idempotent 
	$e' =\mu(xy)$, the group $H'$, and its inverse image $\cH'$ instead of $e,H,\cH$. 
	
	In order to simplify the notation 
	we rename $e',H',\cH'$ as $e,H,\cH$.
	The difference is that, now, we have $(g,q_0)\cdot e = (g,q_0)$ and $\mu(xy) =e$ with $xy \in K^+$.
	Consider $(g,q) \in {\wt Q}_e$ such that $q \neq \bot$ and hence, $q$ is not the sink state of $Q$. Then there exist words $u, v \in A^*$ such that $q_0 \cdot u = q$ and $q\cdot v \in F$. Since $(g,q) = (g,q_0)\cdot u \in {\wt Q}_e$, we obtain $(g,q_0) \cdot u (xy)^d v = (g,q) \cdot v = (g', q_0)$ for some $g' \in G$. Consequently, $u (xy)^d v \in K^*$ and, by synchronization delay, we obtain $u (xy)^d \in K^*$. In particular, 
	$(g,q_0) \cdot u(xy)^d = (g, q_0)$. Thus, $(g,q) = (g,q) \cdot (xy)^d = (g,q_0) u(xy)^d = (g, q_0)$ and therefore, $q = q_0$. 
	Thus, 
	$${\wt Q}_e \subseteq\set{(g,q_0)}{g \in G} \cup \set{(g,\bot)}{g \in G}.$$
	This implies $\cH \subseteq K^*$ by the definition of the automaton.
	(The group $H$ acts trivially on $\set{(g,\bot)}{g \in G}$ and this part is irrelevant in the following.)

	Consider the mapping $\pi : H \to G$ given by $\pi(\mu(u)) = \gamma(u)$ for $u \in \cH$. 
	This mapping is well-defined, since $(g,q_0)\cdot \mu(u) = (g\cdot \gamma(u), q_0)$ for some $(g,q_0) \in {\wt Q}_e$. 
	Thus, the \homo $\gam:\cH \to G$ factorizes as follows:
	$$\gam: \cH \arc{\mu} H \arc{\pi} G.$$
	%Since both $\mu$ and $\gamma$ are homomorphisms, this implies that $\pi$ is a homomorphism. 

	Let us show that the \homo $\pi$ is injective. We know that $H$ acts faithfully on ${\wt Q}_e$. Hence for $h\neq 1$ there is some 
	$(g,q) \in {\wt Q}_e$ such that $(g,q) \cdot h \neq (g,q)$.
	Thus, $(g,q) = (g,q_0)$ and therefore, 
	$$(g,q) \cdot h = (g \pi(h),q_0) \neq (g,q_0).$$
	This shows, as desired, $\pi(h) \neq 1$ and $H$ is a subgroup of $G$. 
	
	{\bf Case 2.} For every state $(g,p)\in {\wt Q}_e$ and every $uv \in \cH$ we have 
	$p \cdot u \notin F$.  Thus, for all $(g,p)\in {\wt Q}_e$ and all $u \in \cH$ we have 
	$$(g,p) \cdot \mu(u) = (g,p \cdot u) = (g,p \cdot \sig(u)).$$
	This means that $H$ acts faithfully on the following set
	$$Q' = \set{p \in Q}{(g,p)\in {\wt Q}_e}.$$
	Let $S'$ denote the submonoid 
	$S'= \set{s\in S}{Q'\cdot s \sse Q'}$, then $\sig(\cH) \sse S'$ and
	$H$ becomes a quotient of $S'$ and therefore, a divisor of $S$. 
	This concludes the proof.
\end{proof}

Next, we introduce a variant of Sch\"utzenberger products to give a short proof  that  the concatenation product of two languages does not introduce new groups, \cite{sch65sf}. 
Let $M$ be a finite monoid and $\varphi : A^* \to M$ be a homomorphism. We define the set
\[
[w] = \set{(\varphi(w_1),\varphi(w_2)) \in M\times M}{w=w_1w_2}.
\]
Further, we define the operations 
\begin{align*}
u \cdot [w] &= \set{(\varphi(u)m,n)}{(m,n) \in [w]}\\
[w] \cdot u &= \set{(m,n\varphi(u))}{(m,n) \in [w]}.
\end{align*}
One can check that $u\cdot [v] \cup [u]\cdot v = [uv]$.
Our variant of the Sch\"utzenberger product is defined as the monoid 
\[
\tilde{M} = \set{[w] \in M \times M}{w\in A^*}
\]
equiped with the operation $[u][v] = [uv]$. This is well-defined since $[u] = [v]$ implies $\varphi(u) = \varphi(v)$. In fact, $\tilde{\varphi} : \tilde{M} \to M$ given by $\tilde{\varphi}([w]) = \varphi(w)$ is a homomorphism. 
It is fairly easy to see that $\tilde{M}$ recognizes the concatenation product over $A^\infty$ as well, see \cite[Proposition 11.7.10]{dr95}.
\begin{proposition}
	Let $L \subseteq A^*$ and $K \subseteq A^\infty$ be languages recognized by $\varphi : A^* \to M$. Then $L \cdot K$ is recognized by the homomorphism $\psi : A^* \to \tilde{M}$ given by $\psi(w) = [w]$.
\end{proposition}
\begin{proof}
	Let $u = u_1u_2 \in A^*$ such that $u_1 \in L$ and $u_2 \in K$ and consider some word $v \in A^*$ such that $\psi(u) = \psi(v)$. 
	Since $(\varphi(u_1),\varphi(u_2)) \in [u] = [v]$, there exists a decomposition $v = v_1v_2$ such that $(\varphi(u_1),\varphi(u_2)) = (\varphi(v_1),\varphi(v_2))$. Consequently, $v_1 \in L$ and $v_2 \in K$, i.e., $v \in L\cdot K$. 
	
	In the case of infinite words let $u = u_1 u_2 \ldots \in L\cdot K$ and $v = v_1 v_2 \ldots$ such that $\psi(u_i) = \psi(v_i)$ for all $i\in \mathbb N$, i.e., $u \sim_\psi v$. We may assume that $u_1 = u'u''$ such that $u' \in L$ and $u'' u_2 \ldots \in K$. Again, there must exist a factorization $v_1 = v'v''$ such that $\varphi(u') = \varphi(v')$ and $\varphi(u'') = \varphi(v'')$. In particular, $v' \in L$. Since $\psi(u_i) = \psi(v_i)$ implies $\varphi(u_i) = \varphi(v_i)$, this yields $(u''u_2) u_3 \cdots \sim_\varphi (v''v_2) v_3 \cdots$ and therefore $v'' v_2 v_3 \cdots \in K$. Thus, $v \in L\cdot K$, which completes the proof.
\end{proof}

We show that every group contained in $\tilde{M}$ is a group in $M$. The argument is a slight deviation of the original argument of Sch\"utzenberger and Petrone \cite[Remark 2]{sch65sf}, in order to adapt to our variant of the Sch\"utzenberger product.
\begin{proposition}
	Let $\varphi : A^* \to M$ be a homomorphism and $\tilde M$ be the corresponding Sch\"utzenberger product. Every group $G \subseteq \tilde M$ can be embedded into $M$.
\end{proposition}
\begin{proof}
	Let $[e]$ be the unit in $G$. Consider again the homomorphism $\tilde{\varphi} : \tilde{M} \to M$. Since $G$ is finite, the set $N =  \set{[w]\in G}{\tilde{\varphi}([w]) = \varphi(w) = \varphi(e) = \tilde{\varphi}([e])}$ is a subgroup of $N$. In fact, $N$ is normal and $G/N$ is isomorphic to $\tilde{\varphi}(G)$, which is a group in $M$. Thus, it remains to show $N = \oneset{[e]}$, i.e., $\tilde{\varphi}$ is injective on $G$. 
	
	Let $[s] \in N$ be an arbitrary element and $[t] \in N$ be its inverse. Then, the following equations holds:
	\begin{multicols}{3}
		\begin{itemize}
			\item $[e]^2 = [e]$
			\item $[e] = [s][t]$
			\item $[s] = [e][s][e]$
		\end{itemize}
	\end{multicols}
	By the first equation we have $[e] = e[e] \cup [e]e$. 
	
	By the second equation and $\varphi(s) = \varphi(t) = \varphi(e)$, it holds $[e] = s[t] \cup [s]t = e[t] \cup [s]e$. 
	Since $e[e] \subseteq [e]$, we conclude $e[s]e \subseteq [e]$.
	Finally, using the third equation, we obtain $$[s] = e([s][e]) \cup [e]se = e(s[e]\cup [s]e) \cup [e]e = e[e] \cup e[s]e \cup [e]e = [e] \cup e[s]e = [e].$$
\end{proof}
\end{document}